\documentclass[12pt, 
]{article}

\usepackage{booktabs}

\usepackage[utf8]{inputenc}
\usepackage{amssymb}
\usepackage{amsmath}

\usepackage{geometry}
\newtheorem{theorem}{Theorem}

\newtheorem{definition}{Definition}
\newtheorem{example}{Example}
\newtheorem{lemma}{Lemma}
\newtheorem{remark}{Remark}

\newenvironment{proof}[1][Proof]{\emph{#1.} }{\  \hfill $\square $ \vspace{5 pt}}

\usepackage{natbib}
\usepackage{amssymb}
\usepackage[dvips]{graphics}
\usepackage{amsmath}

\usepackage{mathpazo}
\usepackage{enumerate}

\usepackage{amsfonts}

\usepackage{amssymb}

\usepackage{enumerate}

\usepackage{mathrsfs}

\usepackage[usenames]{color}
\usepackage{cancel}

\usepackage{pgf,tikz}

\usetikzlibrary{decorations.pathreplacing,angles,quotes}
\usetikzlibrary{arrows.meta}
\usetikzlibrary{arrows, decorations.markings}
\tikzset{myptr/.style={decoration={markings,mark=at position 1 with %
       {\arrow[scale=2,>=stealth]{>}}},postaction={decorate}}}

\usepackage{hyperref}
\hypersetup{
    colorlinks,
    citecolor=blue,
    filecolor=black,
    linkcolor=blue,
    urlcolor=blue
}

\usepackage{pgf,tikz}
\usetikzlibrary{arrows}

\usepackage{fancyhdr}

\usepackage{emptypage}

\usepackage[T1]{fontenc}
\DeclareFontFamily{T1}{calligra}{}
\DeclareFontShape{T1}{calligra}{m}{n}{<->s*[1.44]callig15}{}
\DeclareMathAlphabet\mathcalligra   {T1}{calligra} {m} {n}

\usepackage{multirow}
\usepackage{array}
\usepackage{mathtools}
\usepackage{arydshln}
\usepackage{threeparttable}
\usepackage{booktabs,caption}

\usepackage{ifthen}
\newboolean{showcomments}
\setboolean{showcomments}{true}

\newcommand{\pablo}[1]{  \ifthenelse{\boolean{showcomments}}
{\textcolor{green!50!black}{(T: #1)}}{}}
\newcommand{\marcelo}[1]{\ifthenelse{\boolean{showcomments}}
{\textcolor{red}{(M: #1)}}{}}
\newcommand{\agustin}[1]{  \ifthenelse{\boolean{showcomments}}
{\textcolor{blue!50!black}{(T: #1)}}{}}

\usepackage{orcidlink}
\usepackage{marvosym} 

\begin{document}

\title{Dictatorial profiles: a proof of Arrow's theorem\thanks{I thank Jordi Massó, Éric Rémila, and especially Philippe Solal for their helpful comments. 
}}

\author{Agustín G. Bonifacio\thanks{
 GATE Lyon--St-Étienne, Université Jean Monnet, Saint-\'Etienne, France.
} \hspace{0.75 pt} \thanks{Instituto de Matem\'{a}tica Aplicada San Luis, 
Universidad Nacional de San Luis and CONICET, San Luis, Argentina.} \hspace{0.75 pt} \textsuperscript{\Letter}\hspace{0.35 pt} \orcidlink{0000-0003-2239-8673}}

\date{\today}

\maketitle

\begingroup
\renewcommand{\thefootnote}{\Letter}
\footnotetext[0]{\href{mailto:agustinbonifacio@gmail.com}
       {\texttt{agustinbonifacio@gmail.com}}}
\endgroup

\begin{abstract}

We study social welfare rules with strict individual and social preferences. A preference profile is \emph{dictatorial for a rule} if the social preference coincides with some agent’s preference and no agent with a different preference can change the social preference unilaterally. 
We establish that (i) a rule is \emph{dictatorial} if and only if all of its profiles are dictatorial; 
and (ii) if a rule satisfies \emph{unanimity} and \emph{independence (of irrelevant alternatives)}, then all its profiles must be dictatorial. 
These results constitute a proof of Arrow's Impossibility Theorem.

\bigskip

\bigskip


\end{abstract}

\section{Introduction}

Arrow's Impossibility Theorem states that, with finitely many agents, at least three alternatives, and all strict preferences over alternatives admissible, a social welfare rule satisfying \emph{unanimity} and \emph{independence of irrelevant alternatives} must be \emph{dictatorial} \citep{arrow1963social}. In this note, we provide a proof based on the structure that these properties induce at the profile level.

Arrow's theorem draws its force from the reasonableness of the
demands it makes incompatible. Returning to its proof is a way
of understanding how principles we have good reasons to accept
can lead to a conclusion we have good reasons to resist. The theorem has been proved in several ways
\citep{barbera1980pivotal,reny2001arrow,
geanakoplos2005three,cato2013alternative}. 
Recent contributions include \citet{dantoni2025condorcet} and
\citet{rajsbaum2026combinatorial}.

Our approach focuses on the notion of  dictatorial profile. 
A preference profile is \emph{dictatorial for a rule} if the social preference coincides with some agent’s preference and no agent with a different preference can change the social preference unilaterally. 
This profilewise notion of dictatorial power is in the spirit of the one introduced for social choice rules in \cite{bonifacio2024trade}. 

The idea behind our proof is the following. First, without assuming \emph{unanimity} or \emph{independence}, we show that a rule is \emph{dictatorial} if and only if all of its profiles are dictatorial. The argument applies directly to any finite number of agents: it identifies a profile at which only one agent reports the social preference and shows that this agent's preference must be selected at every profile. Then, we establish that, under \emph{independence}, any non-dictatorial profile can be transformed into a profile that violates \emph{unanimity}. The construction preserves the original social preference throughout. It first eliminates all individual reports equal to that preference, then rearranges the reports so that all agents agree on a comparison opposite to the social one. Together, these results yield Arrow's theorem.

The remainder of the paper is organized as follows. Section \ref{preliminaries} introduces the framework and the standard properties of social welfare rules. Section \ref{profile structure} presents the profilewise notion of  dictatorial power and proves Arrow's Impossibility Theorem. Section \ref{section concluding} gathers some final remarks.

\section{Preliminaries}\label{preliminaries}

There is a set of \textit{agents }$N=\{1,\ldots ,n\}$, with $n\geq 2$, and a finite set of alternatives $X$, with  $|X|=m \geq 3$. Each agent $i\in N$ has a \textit{preference} $P_{i}$ over $X$, which is a strict linear order. The best alternative according to $P_i$ is $t(P_i)$, and the worst is $b(P_i)$. Sometimes, we write $P_i$ as an ordered list: $P_i:x,y,z,\ldots$.  
Let $\mathcal{P}$ be the (universal) domain of all
strict preferences over $X.$  A preference not attached to any agent is denoted by $P_0$. A \textit{(preference) profile} is an ordered list of $n$ preferences, 
 $P=(P_{1},\ldots ,P_{n})\in \mathcal{P}^{n},$  one for each agent. Given $S \subseteq N$, sometimes we write $-S$ as shorthand for $N \setminus S$, and $P_S$ for a subprofile of preferences for agents in $S$. Given $P \in \mathcal{P}^n$,  an expression like $(P_S',P_{-S})$ denotes a new profile assigning $P_i'$ for each $i\in S$.\footnote{As is usual in the literature, when $S=\{j\}$ we simply write $(P'_j,P_{-j})$.} When all agents in $S$ share preference $P_0$ we write $[\,P_0\,]_S$. The restriction of preference $P_i \in \mathcal{P}$ to a subset of alternatives $A \subsetneq  X$ is denoted by $P_i|_A$ and the restriction of a whole profile $P$ is denoted by $P|_A$.

Individual preferences are aggregated into a social preference. A \textit{(social welfare) rule} is a function $f:\mathcal{P}^{n}\longrightarrow \mathcal{P}$ that selects a social preference for each profile of individual preferences. Given $P \in \mathcal{P}^n$, sometimes we write $P_N$ instead of $f(P)$. 

A rule  $f : \mathcal{P}^{n}\longrightarrow \mathcal{P}$ is 
\begin{itemize} 
    \item \textit{unanimous} if, for each pair $x,y \in X$ and each $P \in \mathcal{P}^n$, we have $x f(P) y$ whenever $xP_iy$ for each $i \in N$, 
    \item \emph{independent (of irrelevant alternatives)} if, for each pair of distinct alternatives $x,y \in X$ 
and each pair of profiles $P,P' \in \mathcal{P}^n$ such that $P|_{\{x,y\}}=P'|_{\{x,y\}}$, we have $f(P)|_{\{x,y\}}=f(P')|_{\{x,y\}}$, and 
    \item \emph{dictatorial} if there is $i\in N$ (the \emph{dictator}) such that for each $P\in \mathcal{P}^{n}$, $f(P)=P_i$.
\end{itemize}
The set of all rules is $\mathcal{F}$, the set of all \emph{unanimous} rules is  $\mathcal{U}$, the set of all \emph{independent} rules is $\mathcal{I}$, and the set of all \emph{dictatorial} rules is $\mathcal{D}$. Arrow's impossibility theorem states that $\mathcal{U} \cap \mathcal{I} = \mathcal{D}$.  

\section{Dictatorial profiles and Arrow's theorem}\label{profile structure}

First, we present the central notion studied in the paper.

\begin{definition}\label{def dictatorial}
    Let $f \in \mathcal{F}$ and $P \in \mathcal{P}^n$. Profile $P$ is \textbf{dictatorial for $\boldsymbol{f}$} if there is $i \in N$ such that
\begin{enumerate}[(i)]
    \item $f(P)=P_i$, and 
    \item for each $j \in N\setminus \{i\}$ with $P_j \neq P_i$, we have $f(P_j',P_{-j})=f(P)$ for each $P_j' \in \mathcal{P}$.
\end{enumerate}
Let $\boldsymbol{D_f}$ denote the set of all dictatorial profiles for $f$.
\end{definition}

Our first result states that dictatorial rules are exactly those that have all of their profiles dictatorial. 

\begin{lemma}\label{lema dictatorial}
    $f \in \mathcal{D}$ if and only if $D_f=\mathcal{P}^n$.
\end{lemma}
\begin{proof}
If $f \in \mathcal D$, then $D_f=\mathcal P^n$.
Conversely, suppose that $D_f=\mathcal P^n$.

\medskip

\noindent \textbf{Step 1: there are $\boldsymbol{P^\star\in\mathcal P^n}$ and $\boldsymbol{i^\star\in N}$ such that $$\boldsymbol{f(P^\star)=P_{i^\star}^\star \text{ and } P_j^\star\ne P_{i^\star}^\star \text{ for each }j\in N\setminus\{i^\star\}}.$$}
Choose $P\in\mathcal P^n$ minimizing
$\bigl|\{j\in N:P_j=f(P)\}\bigr|$. This number belongs to
$\{1,\ldots,n\}$ by condition (i) in Definition~\ref{def dictatorial},
so the minimum exists. Let $i\in N$ satisfy $f(P)=P_i$ and set
$N^\star=\{j\in N:P_j=P_i\}$.

Suppose that $|N^\star|>1$. Choose $\widetilde P_0\in\mathcal P\setminus\{P_i\}$ and
replace the preferences outside $N^\star$, one at a time, by
$\widetilde P_0$. Each replacement preserves the outcome by
condition (ii) in Definition~\ref{def dictatorial}. Thus,
\begin{equation}\label{primer perfil}
f([\,P_i\,]_{N^\star},[\,\widetilde P_0\,]_{-N^\star})=P_i.
\end{equation}
Now change the preference of some $k\in N^\star\setminus\{i\}$ to
$\widehat P_k\in\mathcal P\setminus\{P_i,\widetilde P_0\}$.
By condition (i), the new outcome belongs to
$\{P_i,\widetilde P_0,\widehat P_k\}$.
It cannot be $\widetilde P_0$: returning $k$'s preference to $P_i$
would preserve that outcome by condition (ii), contradicting
\eqref{primer perfil}.
The remaining outcomes, $P_i$ and $\widehat P_k$, are reported by
$|N^\star|-1$ agents and one agent, respectively. Both numbers are
smaller than $|N^\star|$, contradicting the choice of $P$.
Hence, $N^\star=\{i\}$. Thus, take $P^\star=P$
and $i^\star=i$.

\medskip

\noindent \textbf{Step 2: $\boldsymbol{i^\star}$ is the dictator.}
Starting from $P^\star$, replace the preferences of all other agents,
one at a time, by arbitrary preferences. Before her replacement,
each agent still reports her original preference, which differs
from $P_{i^\star}^\star$. Condition (ii) therefore gives
\begin{equation}\label{segundo perfil}
f(P_{i^\star}^\star,P_{-i^\star})=P_{i^\star}^\star
\quad\text{for each }P_{-i^\star}\in\mathcal P^{n-1}.
\end{equation}
Let $\overline P_{i^\star}\in\mathcal P\setminus\{P_{i^\star}^\star\}$ and choose
$\overline P_0\in\mathcal P\setminus\{P_{i^\star}^\star,\overline P_{i^\star}\}$.
At $(\overline P_{i^\star},[\,\overline P_0\,]_{-i^\star})$,
condition (i) restricts the outcome to $\overline P_{i^\star}$ or
$\overline P_0$. The latter is impossible: changing $i^\star$'s
preference to $P_{i^\star}^\star$ would preserve $\overline P_0$
by condition (ii), contradicting \eqref{segundo perfil}.
Thus, the outcome is $\overline P_{i^\star}$. Applying the same
replacement argument to this profile gives
\begin{equation}\label{tercer perfil}
f(\overline P_{i^\star},P_{-i^\star})=\overline P_{i^\star}
\quad\text{for each }P_{-i^\star}\in\mathcal P^{n-1}.
\end{equation}
Since $\overline P_{i^\star}$ was arbitrary,
\eqref{segundo perfil} and \eqref{tercer perfil} show that
$i^\star$ is the dictator. Hence, $f \in \mathcal D$.
\end{proof}

To continue, we need some definitions. Let $P_0,P_0', P_0'' \in \mathcal{P}$ and $x,y \in X$. 
First, following \cite{barbera1980pivotal}, two distinct alternatives $x,y$ are  \textbf{contiguous in $\boldsymbol{P_0}$} if there is no alternative $z \in X\setminus \{x,y\}$ such that either  $x\,P_0\, z \,P_0\, y$ or $y \,P_0\, z \,P_0\, x$. Second, 
$P_0'$ is a \textbf{switch of $\boldsymbol{P_0}$} if there is a pair of contiguous alternatives in $P_0$ whose order is reversed at $P_0'$ while all other pairwise comparisons are kept unchanged. Finally, $P_0''$ is the \textbf{contiguous transformation of $\boldsymbol{P_0'}$ with respect to $\boldsymbol{P_0}$} if, whenever $x,y$ are contiguous in $P_0$ and $x \,P_0\, y$, we have that (i) $x \,P_0'\, y$ implies that $x,y$ are contiguous in $P_0''$ and $x \,P_0''\, y$, and (ii) $y \,P_0'\, x$ implies $y \,P_0''\, x$. 

To see that this transformation is well defined, cut the ordered list $P_0$ at each contiguous pair reversed by $P'_0$. Condition (i) requires each resulting segment to remain intact,
while condition (ii) places each segment below the next one. Thus, reversing the order of the segments, while preserving their internal order, gives the unique preference satisfying both conditions.

\begin{example}
     Assume $X=\{x,y,z,t,u,v,w\}$ and consider $P_0, P_0' \in \mathcal{P}$ given by $$P_0: x,y,z,t,u,v,w \quad \text{ and } \quad P_0':z,x,v,t,y,w,u.$$  
     In this case, we obtain segments $[x,y], [z,t,u],[v,w]$. 
     Thus, $P_0'': v,w,z,t,u,x,y$.   \hfill $\Diamond$
\end{example}
If $P'_0\ne P_0$, at least one contiguous comparison in $P_0$
is reversed in $P'_0$. Hence, the construction has at least
two segments and places $b(P_0)$ above $t(P_0)$.

\begin{remark}\label{remark contiguous} 
    Let $P_0,P'_0,P''_0 \in \mathcal{P}$. If $P''_0$ is a contiguous transformation of $P_0'$ with respect to $P_0$, then 
    \begin{enumerate}[(i)]
        \item $x \,P_0''\, y$ if and only if $x \,P_0'\, y$ whenever $x$ and $y$ are contiguous in $P_0$, and
        \item $b(P_0) \,P''_0\, t(P_0)$ when $P_0 \neq P_0'$. 
    \end{enumerate}
\end{remark}

Remember that Arrow's theorem states that $\mathcal U \cap \mathcal I = \mathcal D$. 
It is easy to check that $\mathcal D \subseteq \mathcal U \cap \mathcal I$. The other inclusion is proved next.

\begin{theorem}\label{lema unanime implica dictatorial}
    $\mathcal U \cap \mathcal I \subseteq \mathcal D.$
\end{theorem}
\begin{proof} Let $f \in  \mathcal U \cap \mathcal I$. Assume $f \notin \mathcal D$. Then, by Lemma \ref{lema dictatorial}, $\mathcal{P}^n \setminus D_f \neq \emptyset$. Let $P \notin D_f$ and set $P_N:=f(P)$.

\medskip

\noindent \textbf{Step 1: there is $\boldsymbol{P^\star \in \mathcal{P}^n}$ such that $\boldsymbol{f(P^\star)=P_N}$ and $\boldsymbol{P^\star_i\neq P_N}$ for each $\boldsymbol{i \in N}$.}
Let $N^\star=\{i\in N:P_i=P_N\}$. If $N^\star=\emptyset$, take
$P^\star=P$. Otherwise, condition (i) in Definition~\ref{def dictatorial}
holds, so condition (ii) must fail. Hence, there are
$k\in N\setminus N^\star$ and $P'_k\in\mathcal P$ such that
\begin{equation}\label{eq:step1-deviation}
f(P'_k,P_{-k})\ne P_N.
\end{equation}

\medskip

\noindent \textbf{Claim: there is a switch $\boldsymbol{\widetilde P_0}$ of $\boldsymbol{P_N}$ such that $\boldsymbol{f([\,\widetilde P_0\,]_{N^\star},P_{-N^\star})=P_N}$.}
Assume otherwise. If $\widehat P_0$ is a switch of $P_N$,
\emph{independence} preserves every social comparison except possibly
the switched pair. Thus, the social preference is either $P_N$ or
$\widehat P_0$, and our assumption implies
\begin{equation}\label{eq:step1-switches}
f([\,\widehat P_0\,]_{N^\star},P_{-N^\star})=\widehat P_0
\quad\text{for each switch $\widehat P_0$ of $P_N$.}
\end{equation}

First, we show that
\begin{equation}\label{eq:step1-common}
f([\,P'_0\,]_{N^\star},P_{-N^\star})=P'_0
\quad\text{for each }P'_0\in\mathcal P.
\end{equation}
Let $P_N^{-1}$ be the reverse of $P_N$.\footnote{For any preference
$P_0$, its reverse $P_0^{-1}$ satisfies $x\,P_0^{-1}\,y$ if and only
if $y\,P_0\,x$ for distinct $x,y\in X$.}
For every pair contiguous in $P_N$, $P_N^{-1}$ agrees with the switch
reversing that pair. By \eqref{eq:step1-switches}, \emph{independence}
therefore transfers the reversed social comparison to
$([\,P_N^{-1}\,]_{N^\star},P_{-N^\star})$. Transitivity yields
\begin{equation}\label{eq:step1-reverse}
f([\,P_N^{-1}\,]_{N^\star},P_{-N^\star})=P_N^{-1}.
\end{equation}
On each pair of distinct alternatives, any $P'_0$ agrees with one of
$P_N$ and $P_N^{-1}$. Both common preferences reproduce themselves,
by $f(P)=P_N$ and \eqref{eq:step1-reverse}. Thus, \emph{independence}
gives every comparison of $P'_0$, proving \eqref{eq:step1-common}.

Second, fix $P'_0\in\{P_k,P_k^{-1}\}$ and take a pair contiguous
in $P_k$, labelled $x,y$ so that $x\,P'_0\,y$. Let $\widehat P_k$
be $P_k$ or its switch reversing $\{x,y\}$, chosen to agree with
$P'_k$ on this pair. Choose $z\in X\setminus\{x,y\}$ and
$P_0^+\in\mathcal P$ such that $x\,P_0^+\,z\,P_0^+\,y$.
By \eqref{eq:step1-common},
\[
f([\,P_0^+\,]_{N^\star},P_k,P_{-(N^\star\cup\{k\})})=P_0^+.
\]
Changing $k$'s preference to $\widehat P_k$ preserves her comparisons
on $\{x,z\}$ and $\{z,y\}$. Therefore, \emph{independence} preserves
$x$ above $z$ and $z$ above $y$ socially, and transitivity gives
$x$ above $y$.\footnote{The argument using a third alternative
between $x$ and $y$ also appears in the second proof of
\cite{geanakoplos2005three}.}
Since $P_0^+$ and $P'_0$ agree on $\{x,y\}$, as do $\widehat P_k$
and $P'_k$, \emph{independence} transfers this comparison to
$([\,P'_0\,]_{N^\star},P'_k,P_{-(N^\star\cup\{k\})})$.
These comparisons determine $P'_0$ by transitivity, because $P'_0$
is $P_k$ or its reverse. Hence,
\begin{equation}\label{eq:step1-two-orders}
f([\,P'_0\,]_{N^\star},P'_k,P_{-(N^\star\cup\{k\})})=P'_0
\quad\text{for }P'_0\in\{P_k,P_k^{-1}\}.
\end{equation}
On each pair, $P_N$ agrees with one of these two preferences.
By \eqref{eq:step1-two-orders}, \emph{independence} therefore gives
$f(P'_k,P_{-k})=P_N$, contradicting \eqref{eq:step1-deviation}.
This proves the claim.

\medskip

Take the switch $\widetilde P_0$ supplied by the Claim and set
$P^\star=([\,\widetilde P_0\,]_{N^\star},P_{-N^\star})$.
Then $f(P^\star)=P_N$. Agents in $N^\star$ now report
$\widetilde P_0\ne P_N$, while all other agents already had preferences
different from $P_N$. This completes Step~1.

\medskip

\noindent \textbf{Step 2: 
\emph{unanimity} is violated.} Let $\overline{P} \in \mathcal{P}^n$ be such that, for each $i \in N$,  $\overline{P}_i$ is the contiguous transformation of $P_i^\star$ with respect to $P_N$. First, we claim that \begin{equation}\label{transformada invariante}
    f(\overline{P})=P_N.
\end{equation}
To see this, take $x,y \in X$ contiguous in $P_N$. Since, by Remark \ref{remark contiguous} (i), $x\, \overline{P}_i\, y$ if and only if $x \, P_i^\star \, y$ for each $i \in N$, \emph{independence} implies $x \,f(\overline{P})\, y$ if and only if $x \,f(P^\star)\, y$. Then, by the transitivity of $f(\overline{P})$, we have  $f(\overline{P})=f(P^\star)$. As  $f(P^\star)=P_N$, we get  $f(\overline{P})=P_N$. Second, by Remark \ref{remark contiguous} (ii), $b(P_N) \,\overline{P}_i\, t(P_N)$ for each $i \in N$ whereas $t(P_
N) \, f(\overline{P}) \, b(P_N)$. 
This contradicts \emph{unanimity}. Hence, $f \in \mathcal D$.
\end{proof}

\section{Final Remarks}\label{section concluding}

Dictatorial profiles separate two questions: what makes a rule \emph{dictatorial}, and why \emph{independence} and \emph{unanimity} force it to be so. The first  can be answered without \emph{unanimity} or \emph{independence}: for finitely many agents, a rule is \emph{dictatorial} if all its profiles are dictatorial. For the second, \emph{independence} allows a non-dictatorial profile to be transformed into one violating \emph{unanimity}. Together, these observations yield Arrow's theorem.

The proof relies on strict preferences: comparisons between contiguous alternatives determine the social order and let us control changes in individual preferences. At least three alternatives are needed to  supply the three distinct preferences used in Lemma~\ref{lema dictatorial} and the third alternative placed between an exchanged pair in Step 1 of Theorem~\ref{lema unanime implica dictatorial}.

Lemma~\ref{lema dictatorial} requires a finite set of agents: changing one report need not reduce the cardinality of an infinite set of agents reporting the outcome, and invariance under individual changes need not persist after infinitely many replacements. With finitely many alternatives, the profilewise construction in Theorem~\ref{lema unanime implica dictatorial} applies to any set of agents: \emph{unanimity} and \emph{independence} force every profile to be dictatorial. With infinitely many agents, this property need not imply an individual dictator \citep{kirman1972arrow}.

\bibliographystyle{ecta}
\bibliography{biblioArrow2}

\end{document}